\documentclass[10pt]{article}
\usepackage{amssymb, amsthm}

\newtheorem{theorem}{Theorem}

\newtheorem{lemma}{Lemma}

\usepackage{amsmath}

\usepackage{graphicx}

\usepackage{latexsym}

\title{State Complexity of Two Combined Operations: Reversal-Catenation and Star-Catenation}
\author{Bo Cui, Yuan Gao, Lila Kari, and Sheng Yu}

\begin{document}
\maketitle

\begin{abstract}
In this paper, we show that, due to the structural properties of the resulting automaton obtained from a prior operation, the state complexity of a combined operation may not be equal but close to the mathematical composition of the state complexities of its component operations.
In particular, we provide two witness combined operations: reversal combined with catenation and star combined with catenation.
\end{abstract}

\section{Introduction}
State complexity is a type of descriptional complexity based on {\it
deterministic finite automaton} (DFA) model. The state complexity of
an operation on regular languages is the number of states that are
necessary and sufficient in the worst case for the minimal, complete
DFA that accepts the resulting language of the operation. While many
results on the state complexities of individual operations, such as
union, intersection, catenation, star, reversal, shuffle, orthogonal
catenation, proportional removal, and cyclic
shift~\cite{CaSaYu02,DaDoSa08,Domaratzki02,HoKu02,JiJiSz05,JiOk05,Jriaskova05,SaWoYu04,YuZhSa94,Yu01},
have been obtained in the past 15 years, the research of state
complexities of combined operations, which was initiated by A.
Salomaa, K. Salomaa, and S. Yu in 2007~\cite{SaSaYu07}, is
attracting more attention. This is because, in practice, a
combination of several individual operations, rather than only one
individual operation, is often performed in a certain order. For
example, in order to obtain a precise regular expression, a
combination of basic operations is usually required.

In recent
publications~\cite{CGKY10-cat-sr,CGKY10-cat-ui,EGLY2009,GaSaYu08,GaYu09,GaYu10,JiOk07,LiMaSaYu08,SaSaYu07},
it has been shown that the state complexity of a combined operation
is not always a simple mathematical composition of the state
complexities of its component operations.
% and is more difficult to obtain than that of an individual operation, especially the tight lower bound.
This is sometimes due to the structural properties of the DFA accepting the resulting language obtained from a prior operation of a combined operation.
For example, the languages that are obtained from performing reversal and reach the upper bound of the state complexity of this operation are accepted by DFAs such that half of their states are final; and the initial state of the DFA accepting a language obtained after performing star is always a final state.
As a result, the resulting language obtained from a prior operation may not be among the worst cases of the subsequent operation.
Since such issues are not concerned by the study of the state complexity of individual operations, they are certainly important in the research of the state complexity of combined operations.
Although the number of combined operations is unlimited and it is impossible to study the state complexities of all of them, the study on combinations of two individual operations is clearly necessary.

In this paper, we study the state complexities of reversal combined with catenation, i.e., $L(A)^R L(B)$, and star combined with catenation, i.e., $L(A)^*L(B)$, for minimal complete DFAs $A$ and $B$ of sizes $m,n \ge 1$, respectively.
For $L(A)^R L(B)$, we will show that the general upper bound $\frac{3}{4}2^{m+n}$, which is close to the composition of the state complexities of reversal and catenation $2^{m+n} - 2^{n-1}$, is reachable when $m,n \ge 2$, and it can be lower to $2^{n-1}$ and $2^{m-1}+1$ when $ m = 1$ and $n \ge 1$ and when $m \ge 2$ and $n = 1$, respectively.
For $L(A)^*L(B)$, we will show that, if $A$ has only one final state and it is also the initial state, i.e., $L(A) = L(A)^*$, the state complexity of catenation (also $L(A)^*L(B)$) is $m(2^n-1)-2^{n-1}+1$, which is lower than that of catenation $m2^n - 2^{n-1}$.
In the other cases, that is when $A$ contains some final states that are not the initial state, the state complexity of $L(A)^*L(B)$ is $5 \cdot 2^{m+n-3} - 2^{m-1} - 2^n +1$ instead of $\frac{3}{4}2^{m+n} - 2^{n-1}$, the composition of the state complexities of star and catenation.

In the next section, we introduce the basic definitions and notations used in the paper.
Then, we prove our results on reversal combined with catenation and star combined with catenation in Sections~\ref{sec:rev-cat} and~\ref{sec:star-cat}, respectively.
We conclude the paper in Section~\ref{sec:conclusion}.

\section{Preliminaries}
A DFA is denoted by a 5-tuple $A = (Q,
\Sigma, \delta, s, F)$, where $Q$ is the finite set of states, $\Sigma$ is the finite input alphabet, $\delta: Q \times \Sigma \rightarrow Q$ is the state transition function, $s \in Q$ is the initial state, and $F \subseteq Q$ is the set of final states.
A DFA is said to be complete if $\delta(q,a)$ is defined for all $q \in Q$ and $a \in \Sigma$.
All the DFAs we mention in this paper are assumed to be complete.
We extend $\delta$ to $Q \times \Sigma^* \rightarrow Q$ in the usual way.

A {\it non-deterministic finite automaton} (NFA) is denoted by a 5-tuple $A = (Q,
\Sigma, \delta, s, F)$, where the definitions of $Q$, $\Sigma$, $s$, and $F$ are the same to those of DFAs, but the state transition function $\delta$ is defined as $\delta: Q \times \Sigma \to 2^Q$, where $2^Q$ denotes the power set of $Q$, i.e. the set of all subsets of $Q$.

In this paper, the state transition function $\delta$ is often extended to $\hat{\delta} : 2^Q \times \Sigma \rightarrow 2^Q$. The function $\hat{\delta}$ is defined by $\hat{\delta}(R,a) = \{\delta(r,a) \mid r \in R\}$, for $R \subseteq Q$ and $a \in \Sigma$.
We just write $\delta$ instead of $\hat{\delta}$ if there is no confusion.

A word $w \in \Sigma^*$ is accepted by a finite automaton if $\delta(s,w) \cap F
\neq \emptyset$.
Two states in a finite automaton $A$ are said to be {\it equivalent} if and only if for every word $w \in \Sigma^*$, if $A$ is started in either state with $w$ as input, it either accepts in both cases or rejects in both cases.
It is well-known that a language which is accepted by
an NFA can be accepted by a DFA, and such a language is said to be {\it regular}.
The language accepted by a DFA $A$ is denoted by $L(A)$.
The reader may refer to~\cite{HoMoUl01,Yu97} for more
details about regular languages and finite automata.

The {\it state complexity} of a regular language $L$, denoted by
$sc(L)$, is the number of states of the minimal complete DFA that
accepts $L$. The state complexity of a class $S$ of regular
languages, denoted by $sc(S)$, is the supremum among all $sc(L)$, $L
\in S$. The state complexity of an operation on regular languages is
the state complexity of the resulting languages from the operation as a function of the state complexity of the operand languages.
Thus, in a certain sense, the state complexity of an operation is a worst-case complexity.

\section{Reversal combined with catenation}\label{sec:rev-cat}
In this section, we study the state complexity of $L_1^R L_2$ for an
$m$-state DFA language $L_1$ and an $n$-state DFA language $L_2$. We
first show that the state complexity of $L_1^R L_2$ is upper bounded
by $\frac{3}{4}2^{m+n}$ in general (Theorem~\ref{L_1^R L_2 upper
bound}). Then we prove that this upper bound can be reached when
$m,n \ge 2$ (Theorem~\ref{L_1^R L_2 lower bound}). Next, we
investigate the case when $m=1$ and $n\ge 1$ and prove the state
complexity can be lower to $2^{n-1}$ in such a case
(Theorem~\ref{L_1^R L_2 state complexity m=1 n>=1}). Finally, we
show that the state complexity of $L_1^R L_2$ is $2^{m-1}+1$ when $m
\ge 2$ and $n=1$ (Theorem~\ref{L_1^R L_2 state complexity m>=2
n=1}).

Now, we start with a general upper bound of state complexity of
$L_1^R L_2$ for any integers $m,n \ge 1$.
\begin{theorem}
\label{L_1^R L_2 upper bound} For two integers $m,n \ge 1$, let
$L_1$ and $L_2$ be two regular languages accepted by an $m$-state
DFA and an $n$-state DFA, respectively. Then there exists a DFA of
at most $\frac{3}{4}2^{m+n}$ states that accepts $L_1^R L_2$.
\end{theorem}
\begin{proof}
Let $M=(Q_M,\Sigma , \delta_M , s_M, F_M)$ be a DFA of $m$ states,
$k_1$ final states and $L_1=L(M)$. Let $N=(Q_N,\Sigma , \delta_N ,
s_N, F_N)$ be another DFA of $n$ states and $L_2=L(N)$.

Let $M'=(Q_M,\Sigma , \delta_{M'} , F_M, \{s_M\})$ be an NFA with
$k_1$ initial states. $\delta_{M'}(p,a)=q$ if $\delta_M(q,a)=p$
where $a\in \Sigma$ and $p,q\in Q_M$. Clearly,
$$L(M')=L(M)^R=L_1^R.$$

By performing subset construction on NFA $M'$, we can get an
equivalent, $2^m$-state DFA $A=(Q_A,\Sigma , \delta_A , s_A, F_A)$
such that $L(A)=L_1^R$. Since $M'$ has only one final state $s_M$,
we know that $F_A=\{i\mid i\subseteq Q_M, s_M\in i\}$. Thus, $A$ has
$2^{m-1}$ final states in total. Now we construct a DFA
$B=(Q_B,\Sigma , \delta_B , s_B, F_B)$ accepting the language $L_1^R
L_2$, where
\begin{eqnarray*}
Q_B & = & \{\langle i,j \rangle \mid i\in Q_A\mbox{, } j\subseteq Q_N\},\\
s_B & = & \langle s_A,\emptyset \rangle, \mbox{ if } s_A \not\in F_A;\\
    & = & \langle s_A, \{s_N\} \rangle, \mbox{ otherwise}, \\
F_B & = & \{\langle i,j \rangle\in Q_B \mid j\cap F_N\neq \emptyset\},\\
\delta_B(\langle i,j \rangle, a) & = & \langle i',j' \rangle \mbox{,
if } \delta_A(i,a)=i'\mbox{, }\delta_N(j,a)=j'\mbox{, }a\in
\Sigma \mbox{, } i'\notin F_A; \\
& = &  \langle i',j'\cup \{s_N\} \rangle \mbox{, if }
\delta_A(i,a)=i'\mbox{, }\delta_N(j,a)=j'\mbox{, }a\in \Sigma
\mbox{, } i'\in F_A.
\end{eqnarray*}
From the above construction, we can see that all the states in $B$
starting with $i\in F_A$ must end with $j$ such that $s_N \in j$.
There are in total $2^{m-1}\cdot 2^{n-1}$ states which don't meet
this.

Thus, the number of states of the minimal DFA accepting $L_1^RL_2$
is no more than
$$2^{m+n}-2^{m-1}\cdot 2^{n-1}=\frac{3}{4}2^{m+n}.$$
\end{proof}

This result gives an upper bound for the state complexity of $L_1^R
L_2$. Next we show that this bound is reachable when $m,n \ge 2$.

\begin{theorem}
\label{L_1^R L_2 lower bound} Given two integers $m,n\geq 2$, there
exists a DFA $M$ of $m$ states and a DFA $N$ of $n$ states such that
any DFA accepting $L(M)^R L(N)$ needs at least $\frac{3}{4}2^{m+n}$
states.
\end{theorem}
\begin{proof}
Let $M=(Q_M,\Sigma , \delta_M , 0, \{m-1\} )$ be a DFA, shown in
Figure~\ref{DFAM-rev-cat}, where $Q_M = \{0,1,\ldots ,m-1\}$,
$\Sigma = \{a,b,c,d\}$, and the transitions are given as:
\begin{itemize}
\item $\delta_M(i, a) = i+1 \mbox{ mod }m \mbox{, } i=0, \ldots , m-1,$
\item $\delta_M(i, b) = i \mbox{, } i=0, \ldots , m-2,$ $\delta_M(m-1, b) = m-2 \mbox{, }$
\item $\delta_M(m-2, c) = m-1 \mbox{, }$ $\delta_M(m-1, c) = m-2 \mbox{,}$\\
 if $m\geq 3$, $\delta_M(i, c) = i \mbox{, } i=0, \ldots , m-3,$
\item $\delta_M(i, d) = i \mbox{, } i=0, \ldots , m-1,$
\end{itemize}
\begin{figure}[ht]
  \begin{center}
  \includegraphics[scale=0.17]{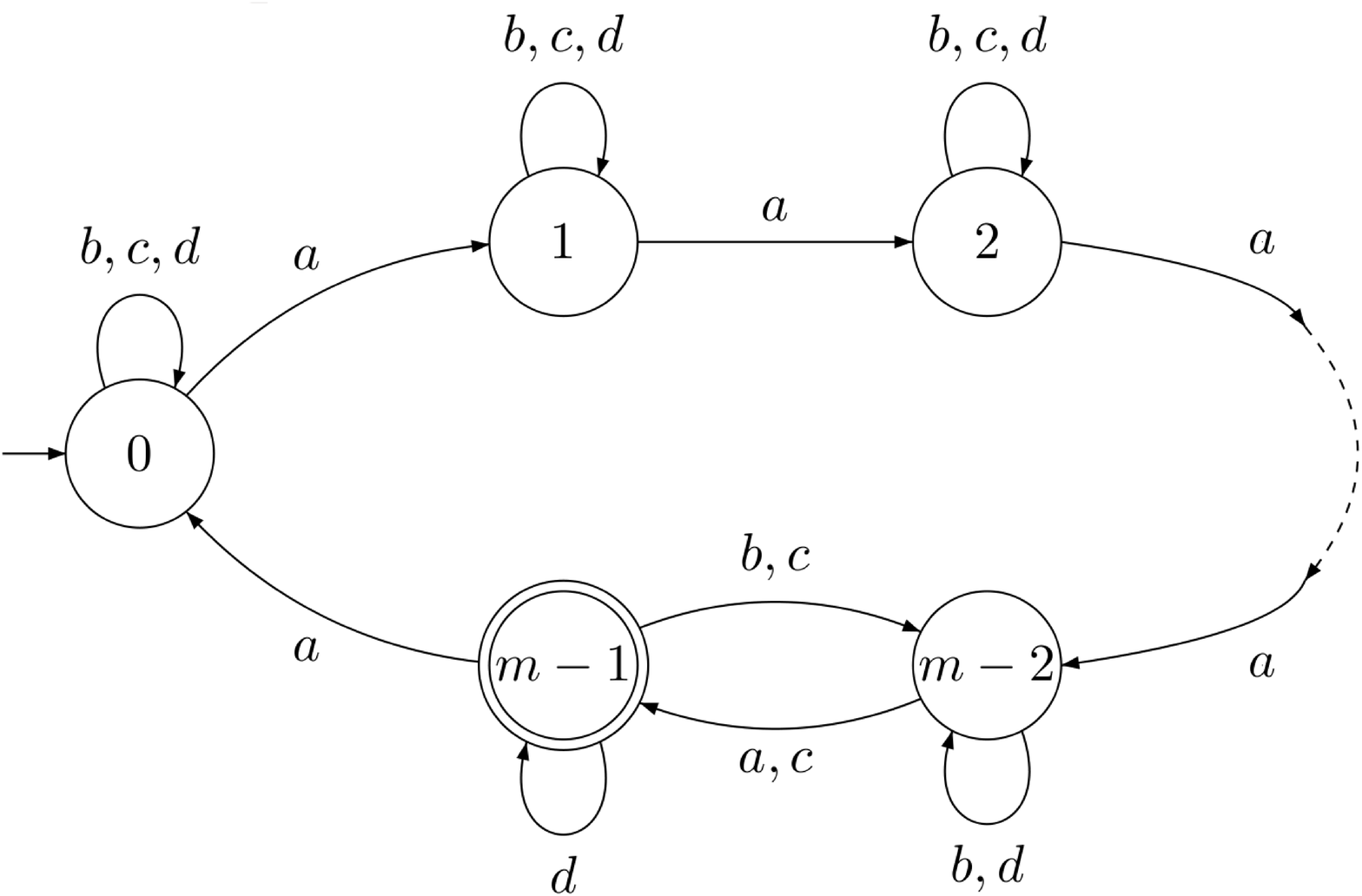}
  \end{center}
  \caption{Witness DFA $M$ of Theorem~\ref{L_1^R L_2 lower bound} showing that the upper bound in Theorem~\ref{L_1^R L_2 upper bound} is reachable when $m, n \ge 2$}
\label{DFAM-rev-cat}
\end{figure}

Let $N=(Q_N,\Sigma , \delta_N , 0, \{n-1\} )$ be a DFA, shown in
Figure~\ref{DFAN-rev-cat}, where $Q_N = \{0,1,\ldots ,n-1\}$,
$\Sigma = \{a,b,c,d\}$, and the transitions are given as:
\begin{itemize}
\item $\delta_N(i, a) = i \mbox{, } i=1, \ldots , n-1,$
\item $\delta_N(i, b) = i \mbox{, } i=1, \ldots , n-1,$
\item $\delta_N(i, c) = 0 \mbox{, } i=1, \ldots , n-1,$
\item $\delta_N(i, d) = i+1 \mbox{ mod }n \mbox{, } i=0, \ldots , n-1,$
\end{itemize}
\begin{figure}[ht]
  \begin{center}
  \includegraphics[scale=0.17]{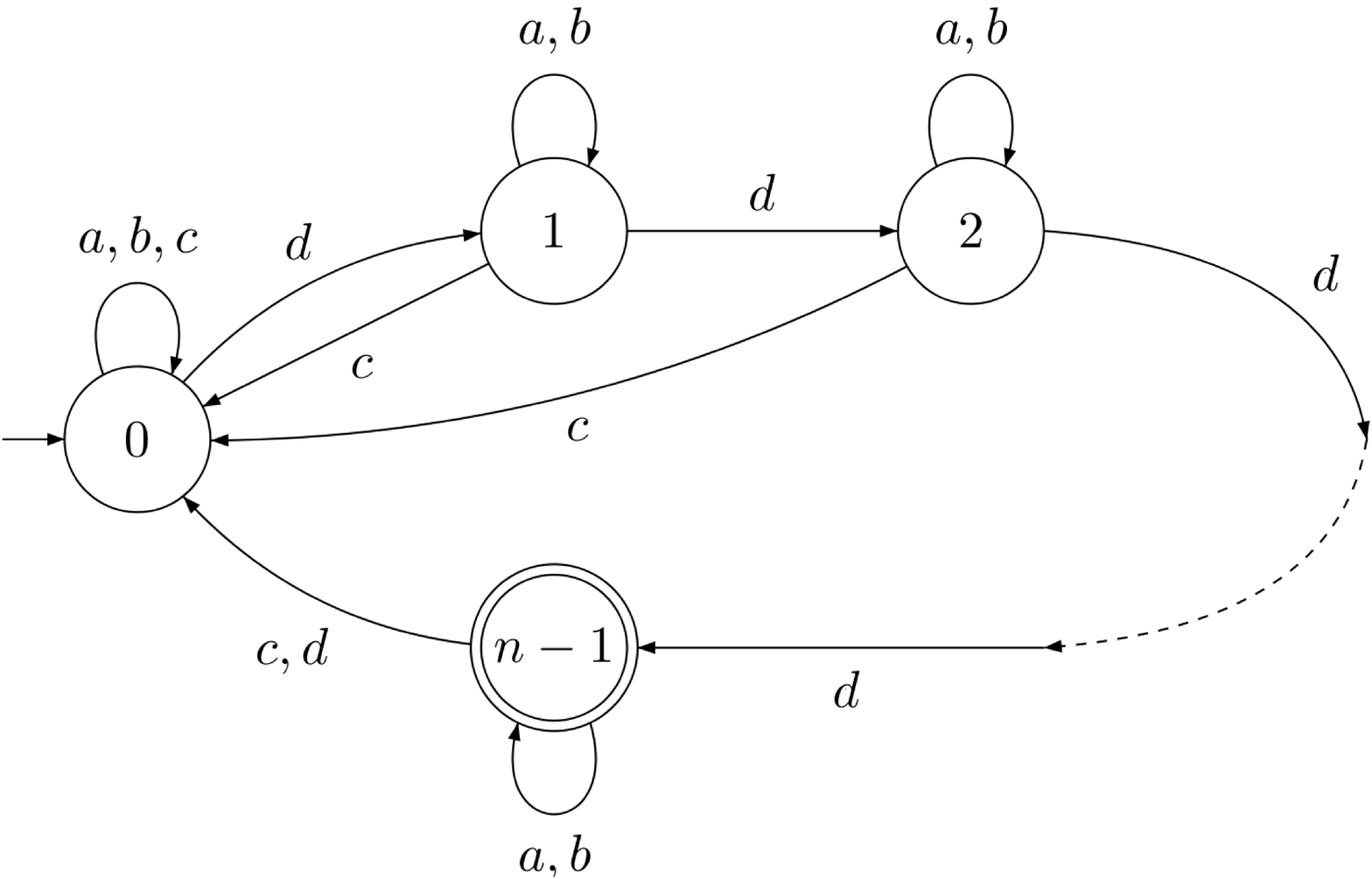}
  \end{center}
  \caption{Witness DFA $N$ of Theorem~\ref{L_1^R L_2 lower bound} showing that the upper bound in Theorem~\ref{L_1^R L_2 upper bound} is reachable when $m, n \ge 2$}
\label{DFAN-rev-cat}
\end{figure}
Now we design a DFA $A=(Q_A, \Sigma , \delta_A , \{m-1\}, F_A )$,
where $Q_A = \{q \mid q\subseteq Q_M\}$, $\Sigma = \{a,b,c,d\}$,
$F_A = \{q \mid 0\in q \mbox{, }q\in Q_A\}$, and the transitions are
defined as:
\[
\delta_A(p, e) = \{j \mid \delta_M(j, e)=i\mbox{, }i\in p\} \mbox{,
} p\in Q_A\mbox{, } e \in \Sigma.
\]
It is easy to see that $A$ is a DFA that accepts $L(M)^R$. We prove
that $A$ is minimal before using it.

(I) We first show that every state $I\in Q_A$, is reachable from
$\{m-1\}$. There are three cases.

\begin{itemize}
\item[{\rm 1.}]$|I|=0$.
$|I|=0$ if and only if $I=\emptyset$. $\delta_A(\{ m-1 \}, b) =
I=\emptyset.$
\item[{\rm 2.}]$|I|=1$.
Let $I=\{ i \}$, $0\leq i\leq m-1$. $\delta_A(\{ m-1 \}, a^{m-1-i})
=I.$
\item[{\rm 3.}]$2\leq |I|\leq m$.
Let $I=\{ i_1, i_2, \ldots ,i_k \}$, $0\leq i_1<i_2< \ldots <i_k
\leq m-1$, $2\leq k\leq m$. $\delta_A(\{ m-1 \}, w) = I$, where
$$w = ab(ac)^{i_2-i_1-1}ab(ac)^{i_{3}-i_{2}-1}\cdots ab(ac)^{i_k-i_{k-1}-1}a^{m-1-i_k}.$$
\end{itemize}

(II) Any two different states $I$ and $J$ in $Q_A$ are
distinguishable.

Without loss of generality, we may assume that $|I|\geq |J|$. Let
$x\in I-J$. Then a string $a^{x}$ can distinguish these two states
because
\begin{eqnarray*}
\delta_A(I, a^{x})& \in & F_A,\\
\delta_A(J, a^{x}) & \notin & F_A.
\end{eqnarray*}

Due to (I) and (II), $A$ is a minimal DFA with $2^m$ states which
accepts $L(M)^R$. Now let $B=(Q_B, \Sigma , \delta_B , s_B, F_A\}$
be another DFA, where
\begin{eqnarray*}
Q_B & = & \{\langle p,q\rangle  \mid p\in Q_A-F_A \mbox{, }q\subseteq Q_N\}\\
& & \qquad \cup \,\, \{\langle p', q'\rangle \mid p'\in F_A \mbox{, }q'\subseteq Q_N\mbox{, }0\in q'\},\\
\Sigma & = & \{a,b,c,d\},\\
s_B & = & \langle \{m-1\},\emptyset \rangle,\\
F_B & = & \{\langle p,q\rangle \mid n-1\in q \mbox{, } \langle
p,q\rangle\in Q_B\},
\end{eqnarray*}
and for each state $\langle p,q\rangle\in Q_B$ and each letter $e\in
\Sigma,$

\begin{eqnarray*}
\delta_B(\langle p,q\rangle, e) = \left\{
\begin{array}{l l}
  \langle p',q'\rangle & \mbox{if }\delta_A(p, e)=p'\notin F_A\mbox{, } \delta_N(q, e)=q',\\
  \langle p',q'\rangle &  \mbox{if }\delta_A(p, e)=p'\in F_A\mbox{, }\delta_N(q, e)=r'\mbox{, $q' =r'\cup \{0\}$.}  \\
%   & \qquad \delta_N(q, e)=r'\mbox{, $q' =r'\cup \{0\}$,}\\
\end{array} \right.
\end{eqnarray*}
As we mentioned in last proof, all the states starting with $p\in
F_A$ must end with $q\subseteq Q_N$ such that $0\in q$. Clearly, $B$
accepts the language $L(M)^RL(N)$ and it has
$$2^m\cdot 2^n-2^{m-1}\cdot 2^{n-1}=\frac{3}{4}2^{m+n}$$
states. Now we show that $B$ is a minimal DFA.

(I) Every state $\langle p,q\rangle \in Q_B$ is reachable. We
consider the following five cases:

\begin{itemize}

\item[{\rm 1.}]$p=\emptyset$, $q=\emptyset$.
$\langle \emptyset,\emptyset\rangle$ is the sink state of $B$.
$\delta_B(\langle \{m-1\},\emptyset \rangle, b) = \langle
p,q\rangle$.

\item[{\rm 2.}]$p\neq \emptyset$, $q=\emptyset$.
Let $p=\{ p_1, p_2, \ldots ,p_k \}$, $1\leq p_1<p_2< \ldots <p_k
\leq m-1$, $1\leq k\leq m-1$. Note that $0\notin p$, because $0\in
p$ guarantees $0\in q$. $\delta_B(\langle \{m-1\},\emptyset \rangle,
w) = \langle p,q\rangle$, where
$$w = ab(ac)^{p_2-p_1-1}ab(ac)^{p_{3}-p_{2}-1}\cdots ab(ac)^{p_k-p_{k-1}-1}a^{m-1-p_k}.$$
Please note that $w=a^{m-1-p_1}$ when $k=1$.

\item[{\rm 3.}]$p= \emptyset$, $q\neq \emptyset$.
In this case, let $q=\{ q_1, q_2, \ldots ,q_l \}$, $0\leq q_1<q_2<
\ldots <q_l \leq n-1$, $1\leq l\leq n$. $\delta_B(\langle
\{m-1\},\emptyset \rangle, x) = \langle p,q\rangle$, where
$$x = a^md^{q_l-q_{l-1}}a^md^{q_{l-1}-q_{l-2}}\cdots a^md^{q_2-q_1}a^md^{q_1}b.$$

\item[{\rm 4.}]$p\neq \emptyset$, $0\notin p$, $q\neq \emptyset$.
Let $p=\{ p_1, p_2, \ldots ,p_k \}$, $1\leq p_1<p_2< \ldots <p_k
\leq m-1$, $1\leq k\leq m-1$ and $q=\{ q_1, q_2, \ldots ,q_l \}$,
$0\leq q_1<q_2< \ldots <q_l \leq n-1$, $1\leq l\leq n$. We can find
a string $uv$ such that $\delta_B(\langle \{m-1\},\emptyset \rangle,
uv) = \langle p,q\rangle$, where
$$u = ab(ac)^{p_2-p_1-1}ab(ac)^{p_{3}-p_{2}-1}\cdots ab(ac)^{p_k-p_{k-1}-1}a^{m-1-p_k},$$
$$v = a^md^{q_l-q_{l-1}}a^md^{q_{l-1}-q_{l-2}}\cdots a^md^{q_2-q_1}a^md^{q_1}.$$

\item[{\rm 5.}]$p\neq \emptyset$, $0\in p$, $m-1\notin p$, $q\neq \emptyset$.
Let $p=\{ p_1, p_2, \ldots ,p_k \}$, $0= p_1<p_2< \ldots <p_k <m-1$,
$1\leq k\leq m-1$ and $q=\{ q_1, q_2, \ldots ,q_l \}$, $0= q_1<q_2<
\ldots <q_l \leq n-1$, $1\leq l\leq n$. Since $0$ is in $p$,
according to the definition of $B$, $0$ has to be in $q$ as well.
There exists a string $u'v'$ such that $\delta_B(\langle
\{m-1\},\emptyset \rangle, u'v') = \langle p,q\rangle$, where
$$u' = ab(ac)^{p_2-p_1-1}ab(ac)^{p_{3}-p_{2}-1}\cdots ab(ac)^{p_k-p_{k-1}-1}a^{m-2-p_k},$$
$$v' = a^md^{q_l-q_{l-1}}a^md^{q_{l-1}-q_{l-2}}\cdots a^md^{q_2-q_1}a^md^{q_1}a.$$

\item[{\rm 6.}]$p\neq \emptyset$, $\{0,m-1\}\subseteq p$, $q\neq \emptyset$.
Let $p=\{ p_1, p_2, \ldots ,p_k \}$, $0= p_1<p_2< \ldots <p_k =m-1$,
$2\leq k\leq m$ and $q=\{ q_1, q_2, \ldots ,q_l \}$, $0= q_1<q_2<
\ldots <q_l \leq n-1$, $1\leq l\leq n$. In this case, we have
\begin{eqnarray*}
\langle p,q\rangle = \left\{
\begin{array}{l l}
  \delta_B(\langle \{0,1,p_2+1,\ldots,p_{k-1}+1\} , q \rangle, a), & \mbox{if }m-2\notin p,\\
  \delta_B(\langle p-\{m-1\},q\rangle, b), & \mbox{if }m-2\in p,
\end{array} \right.
\end{eqnarray*}
where states $\langle \{0,1,p_2+1,\ldots,p_{k-1}+1\} , q \rangle$
and $\langle p-\{m-1\},q\rangle$ have been proved to be reachable in
Case 5.

%\item[{\rm 6.}]$p\neq \emptyset$, $0$, $m-1\in p$, $m-2\notin p$, $q\neq \emptyset$.
%Let $p=\{ p_1, p_2, \ldots ,p_k \}$, $0= p_1<p_2< \ldots <p_k =m-1$,
%$2\leq k\leq m$ and $q=\{ q_1, q_2, \ldots ,q_l \}$, $0= q_1<q_2<
%\ldots <q_l \leq n-1$, $1\leq l\leq n$. In this case, $p_{k-1}\neq
%m-2$ and we have
%\[
%\langle \{ 0, p_2, \ldots, p_{k-1}, m-1\},q\rangle = \delta_B
%(\langle \{0,1,p_2+1,\ldots,p_{k-1}+1\} , q \rangle, a),
%\]
%where state $\langle \{0,1,p_2+1,\ldots,p_{k-1}+1\} , q \rangle$ has
%been included in Case 5.

%\item[{\rm 7.}]$p\neq \emptyset$, $0, m-2, m-1\in p$, $q\neq \emptyset$.
%Let $p=\{ p_1, p_2, \ldots ,p_k \}$, $0= p_1<p_2< \ldots <p_k =m-1$,
%$3\leq k\leq m$ and $q=\{ q_1, q_2, \ldots ,q_l \}$, $0= q_1<q_2<
%\ldots <q_l \leq n-1$, $1\leq l\leq n$. $B$ can reach state $\langle
%p,q\rangle$ from state $\langle p-\{m-1\},q\rangle$ by reading a
%letter $b$. We have proved that state $\langle p-\{m-1\},q\rangle$
%is reachable from $\langle \{m-1\},\emptyset \rangle$ in Case 5.

\end{itemize}

(II) We then show that any two different states $\langle
p_1,q_1\rangle$ and $\langle p_2,q_2\rangle$ in $Q_B$ are
distinguishable.

\begin{itemize}

\item[{\rm 1.}]$q_1\neq q_2$.
Without loss of generality, we may assume that $|q_1|\geq |q_2|$. Let $x\in q_1-q_2$. A string $d^{n-1-x}$ can distinguish them because
\begin{eqnarray*}
\delta_B(\langle p_1,q_1\rangle, d^{n-1-x})& \in & F_B,\\
\delta_B(\langle p_2,q_2\rangle, d^{n-1-x}) & \notin & F_B.
\end{eqnarray*}

\item[{\rm 2.}]$p_1\neq p_2$, $q_1= q_2$. Without loss of generality, we assume that $|p_1|\geq |p_2|$. Let $y\in p_1-p_2$. Then there always exists a string $a^yc^2d^{n}$ such that
\begin{eqnarray*}
\delta_B(\langle p_1,q_1\rangle, a^yc^2d^{n})& \in & F_B,\\
\delta_B(\langle p_2,q_2\rangle, a^yc^2d^{n}) & \notin & F_B.
\end{eqnarray*}

\end{itemize}
Since all the states in $B$ are reachable and pairwise
distinguishable, DFA $B$ is minimal. Thus, any DFA accepting
$L(M))^RL(N)$ needs at least $\frac{3}{4}2^{m+n}$ states.
\end{proof}

This result gives a lower bound for the state complexity of
$L_1^RL_2$ when $m,n \ge 2$. It coincides with the upper bound shown
in Theorem~\ref{L_1^R L_2 upper bound} exactly. Thus, we obtain the
state complexity of the combined operation $L_1^RL_2$ for $m \ge 2$
and $n\ge 2$.
\begin{theorem}
\label{L_1^R L_2 state complexity} For any integers $m, n\geq 2$,
let $L_1$ be an $m$-state DFA language and $L_2$ be an $n$-state DFA
language. Then $\frac{3}{4}2^{m+n}$ states are both necessary and
sufficient in the worst case for a DFA to accept $L_1^RL_2$.
\end{theorem}

In the rest of this section, we study the remaining cases when
either $m =1$ or $n=1$.

We first consider the case when $m=1$ and $n\geq 2$. In this case,
$L_1=\emptyset$ or $L_1=\Sigma^*$. $L_1^RL_2=L_1L_2$ holds no matter
$L_1$ is $\emptyset$ or $\Sigma^*$, since $\emptyset ^R=\emptyset$
and $(\Sigma^*)^R=\Sigma^*$. It has been shown in~\cite{YuZhSa94}
that $2^{n-1}$ states are both sufficient and necessary in the worst
case for a DFA to accept the catenation of a 1-state DFA language
and an $n$-state DFA language, $n\ge 2$.

When $m =1$ and $n=1$, it is also easy to see that $1$ state is
sufficient and necessary in the worst case for a DFA to accept
$L_1^RL_2$, because $L_1^RL_2$ is either $\emptyset$ or $\Sigma^*$.
Thus, we have the following theorem concerning the state complexity
of $L_1^RL_2$ for $m=1$ and $n\ge 1$.

\begin{theorem}
\label{L_1^R L_2 state complexity m=1 n>=1} Let $L_1$ be a 1-state
DFA language and $L_2$ be an $n$-state DFA language, $n\ge 1$. Then
$2^{n-1}$ states are both sufficient and necessary in the worst case
for a DFA to accept $L_1^RL_2$.
\end{theorem}

Now, we study the state complexity of $L_1^RL_2$ for $m \ge 2$ and
$n = 1$. Let us start with the following upper bound.

\begin{theorem}
\label{L_1^R L_2 upper bound m>=2 n=1} For any integer $m\ge 2$, let
$L_1$ and $L_2$ be two regular languages accepted by an $m$-state
DFA and a $1$-state DFA, respectively. Then there exists a DFA of at
most $2^{m-1}+1$ states that accepts $L_1^R L_2$.
\end{theorem}

\begin{proof}
Let $M=(Q_M,\Sigma , \delta_M , s_M, F_M)$ be a DFA of $m$ states,
$m\ge 2$, $k_1$ final states and $L_1=L(M)$. Let $N$ be another DFA
of $1$ state and $L_2=L(N)$. Since $N$ is a complete DFA, as we
mentioned before, $L(N)$ is either $\emptyset$ or $\Sigma^*$.
Clearly, $L_1^R\cdot \emptyset=\emptyset$. Thus, we need to consider
only the case $L_2=L(N)=\Sigma^*$.

We construct an NFA $M'=(Q_M,\Sigma , \delta_{M'} , F_M, \{s_M\})$
with $k_1$ initial states which is similar to the proof of
Theorem~\ref{L_1^R L_2 upper bound}. $\delta_{M'}(p,a)=q$ if
$\delta_M(q,a)=p$ where $a\in \Sigma$ and $p,q\in Q_M$. It is easy
to see that
$$L(M')=L(M)^R=L_1^R.$$

By performing subset construction on NFA $M'$, we get an equivalent,
$2^m$-state DFA $A=(Q_A,\Sigma , \delta_A , s_A, F_A)$ such that
$L(A)=L_1^R$. $F_A=\{i\mid i\subseteq Q_M, s_M\in i\}$ because $M'$
has only one final state $s_M$. Thus, $A$ has $2^{m-1}$ final states in total.

Define $B=(Q_B,\Sigma , \delta_B , s_B, \{f_B\})$ where $f_B\notin
Q_A$, $Q_B=(Q_A-F_A)\cup \{f_B\}$,
\begin{eqnarray*}
s_B = \left\{
\begin{array}{l l}
  s_A & \mbox{if }s_A\notin F_A,\\
  f_B & \mbox{otherwise.}\\
\end{array} \right.
\end{eqnarray*}
%$s_B=s_A$ if $s_A\notin F_A$, $s_B=f_B$ otherwise.
and for any $a\in \Sigma$ and $p\in Q_B$,
\begin{eqnarray*}
\delta_B(p, a) = \left\{
\begin{array}{l l}
  \delta_A(p, a) & \mbox{if }\delta_A(p, a)\notin F_A,\\
  f_B & \mbox{if }\delta_A(p, a)\in F_A,\\
  f_B & \mbox{if }p=f_B.\\
\end{array} \right.
\end{eqnarray*}
The automaton $B$ is exactly the same as $A$ except that $A$'s
$2^{m-1}$ final states are made to be sink states and these sink,
final states are merged into one, since they are equivalent. When
the computation reaches the final state $f_B$, it remains there.
Now, it is clear that $B$ has
$$2^m-2^{m-1}+1=2^{m-1}+1$$ states and $L(B)=L_1^R \Sigma^*$.
\end{proof}
This theorem shows an upper bound for the state complexity of $L_1^R
L_2$ for $m \ge 2$ and $n = 1$. Next we prove that this upper bound
is reachable.

\begin{lemma}
\label{L_1^R L_2 lower bound m=2 or 3 n=1}Given an integer $m=2$ or
$3$, there exists an $m$-state DFA $M$ and a $1$-state DFA $N$ such
that any DFA accepting $L(M)^RL(N)$ needs at least $2^{m-1}+1$
states.
\end{lemma}

\begin{proof}
When $m=2$ and $n = 1$. We can construct the following witness DFAs.
Let $M=(\{0, 1\},\Sigma , \delta_M , 0, \{1\} )$ be a DFA, where
$\Sigma = \{a,b\}$, and the transitions are given as:
\begin{itemize}
\item $\delta_M(0, a) = 1 \mbox{, } \delta_M(1, a) = 0,$
\item $\delta_M(0, b) = 0 \mbox{, } \delta_M(1, b) = 0.$
\end{itemize}
Let $N$ be the DFA accepting $\Sigma^*$. Then the resulting DFA for
$L(M)^R \Sigma^*$ is $A=(\{0, 1, 2\},\Sigma , \delta_A , 0, \{1\} )$
where
\begin{itemize}
\item $\delta_A(0, a) = 1 \mbox{, } \delta_A(1, a) = 1\mbox{, }\delta_A(2, a) = 2\mbox{, }$
\item $\delta_A(0, b) = 2 \mbox{, } \delta_A(1, b) = 1\mbox{, }\delta_A(2, b) = 2.$
\end{itemize}

When $m=3$ and $n = 1$. The witness DFAs are as follows. Let
$M'=(\{0, 1, 2\},\Sigma' , \delta_{M'} , 0, \{2\} )$ be a DFA, where
$\Sigma' = \{a,b,c\}$, and the transitions are:
\begin{itemize}
\item $\delta_{M'}(0, a) = 1 \mbox{, } \delta_{M'}(1, a) = 2\mbox{, }\delta_{M'}(2, a) = 0\mbox{, }$
\item $\delta_{M'}(0, b) = 0 \mbox{, } \delta_{M'}(1, b) = 0\mbox{, }\delta_{M'}(2, b) = 1\mbox{, }$
\item $\delta_{M'}(0, c) = 0 \mbox{, } \delta_{M'}(1, c) = 2\mbox{, }\delta_{M'}(2, c) = 1\mbox{. }$
\end{itemize}
Let $N'$ be the DFA accepting $\Sigma'^*$. The resulting DFA for
$L(M')^R \Sigma'^*$ is $A'=(\{0, 1, 2, 3, 4\},\Sigma' , \delta_{A'}
, 0, \{3\} )$ where
\begin{itemize}
\item $\delta_{A'}(0, a) = 1 \mbox{, } \delta_{A'}(1, a) = 3\mbox{, }\delta_{A'}(2, a) = 2\mbox{, }\delta_{A'}(3, a) = 3\mbox{, }\delta_{A'}(4, a) = 3\mbox{, }$
\item $\delta_{A'}(0, b) = 2 \mbox{, } \delta_{A'}(1, b) = 4\mbox{, }\delta_{A'}(2, b) = 2\mbox{, }\delta_{A'}(3, b) = 3\mbox{, }\delta_{A'}(4, b) = 4\mbox{, }$
\item $\delta_{A'}(0, c) = 1 \mbox{, } \delta_{A'}(1, c) = 0\mbox{, }\delta_{A'}(2, c) = 2\mbox{, }\delta_{A'}(3, c) = 3\mbox{, }\delta_{A'}(4, c) = 4\mbox{. }$
\end{itemize}
\end{proof}

The above result shows that the bound $2^{m-1}+1$ is reachable when
$m$ is equal to 2 or 3 and $n = 1$. The last case is $m\ge 4$ and $n =
1$.

\begin{theorem}
\label{L_1^R L_2 lower bound m>=4 n=1} Given an integer $m\ge 4$,
there exists a DFA $M$ of $m$ states and a DFA $N$ of $1$ state such
that any DFA accepting $L(M)^R L(N)$ needs at least $2^{m-1}+1$
states.
\end{theorem}

\begin{proof}
Let $M=(Q_M,\Sigma , \delta_M , 0, \{m-1\} )$ be a DFA,
shown in Figure~\ref{DFAM-rev-cat-n=1},
where $Q_M = \{0,1,\ldots ,m-1\}$, $m\ge 4$, $\Sigma = \{a,b,c,d\}$,
and the transitions are given as:
\begin{itemize}
\item $\delta_M(i, a) = i+1 \mbox{ mod }m \mbox{, } i=0, \ldots , m-1,$
\item $\delta_M(i, b) = i \mbox{, } i=0, \ldots , m-2\mbox{, } \delta_M(m-1, b) = m-2 \mbox{, }$
\item $\delta_M(i, c) = i \mbox{, } i=0, \ldots , m-3\mbox{, } \delta_M(m-2, c) = m-1 \mbox{, } \delta_M(m-1, c) = m-2 \mbox{,}$
\item $\delta_M(0, d) = 0 \mbox{, } \delta_M(i, d) = i+1 \mbox{, } i=1, \ldots , m-2\mbox{, } \delta_M(m-1, d) = 1 \mbox{. }$
\end{itemize}
\begin{figure}[ht]
  \begin{center}
  \includegraphics[scale=0.17]{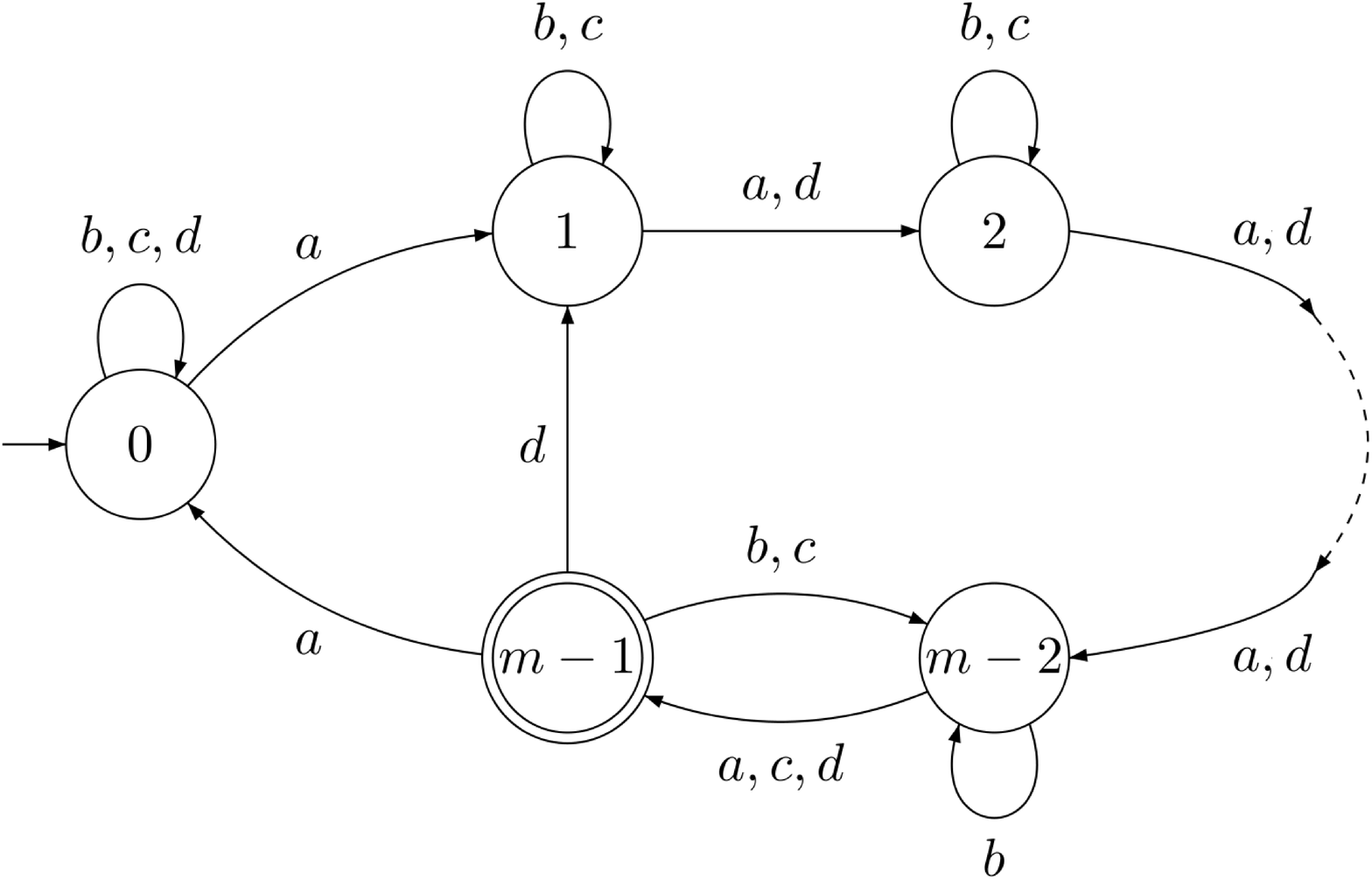}
  \end{center}
  \caption{Witness DFA $M$ of Theorem~\ref{L_1^R L_2 lower bound m>=4 n=1} showing that the upper bound in Theorem~\ref{L_1^R L_2 upper bound m>=2 n=1} is reachable when $m\ge 4$ and $n=1$}
\label{DFAM-rev-cat-n=1}
\end{figure}
Let $N$ be the DFA accepting $\Sigma^*$. Then
$L(M)^RL(N)=L(M)^R\Sigma^*$. Now we design a DFA $A=(Q_A, \Sigma ,
\delta_A , \{m-1\}, F_A)$ similar to the proof of Theorem~\ref{L_1^R
L_2 lower bound}, where $Q_A = \{q \mid q\subseteq Q_M\}$, $\Sigma =
\{a,b,c,d\}$, $F_A = \{q \mid 0\in q \mbox{, }q\in Q_A\}$, and the
transitions are defined as:
\[
\delta_A(p, e) = \{j \mid \delta_M(j, e)=i\mbox{, }i\in p\} \mbox{,
} p\in Q_A\mbox{, } e \in \Sigma.
\]
It is easy to see that $A$ is a DFA that accepts $L(M)^R$. Since the
transitions of $M$ on letters $a$, $b$, and $c$ are exactly the same
as those of DFA $M$ in the proof of Theorem~\ref{L_1^R L_2 lower
bound}, we can say that $A$ is minimal and it has $2^{m}$ states,
among which $2^{m-1}$ states are final.

%(** this definition can be omitted, because it is exactly the same as in the previous theorem **

Define $B=(Q_B,\Sigma , \delta_B , s_B, \{f_B\})$ where $f_B\notin
Q_A$, $Q_B=(Q_A-F_A)\cup \{f_B\}$,
\begin{eqnarray*}
s_B = \left\{
\begin{array}{l l}
  s_A & \mbox{if }s_A\notin F_A,\\
  f_B & \mbox{otherwise.}\\
\end{array} \right.
\end{eqnarray*}
and for any $e\in \Sigma$ and $I\in Q_B$,
\begin{eqnarray*}
\delta_B(I, e) = \left\{
\begin{array}{l l}
  \delta_A(I, e) & \mbox{if }\delta_A(I, e)\notin F_A,\\
  f_B & \mbox{if }\delta_A(I, e)\in F_A,\\
  f_B & \mbox{if }I=f_B.\\
\end{array} \right.
\end{eqnarray*}
DFA $B$ is the same as $A$ except that $A$'s $2^{m-1}$ final states
are changed into sink states and merged to one sink, final state, as
we did in the proof of Theorem~\ref{L_1^R L_2 upper bound m>=2 n=1}.
%(** new **
%Let $B = (Q_B,\Sigma , \delta_B , s_B, \{f_B\})$ be the DFA constructed from $M$ exactly as described in the proof of the previous theorem.
%****)
Clearly, $B$ has $2^m-2^{m-1}+1=2^{m-1}+1$ states and
$L(B)=L(M)^R\Sigma^*$. Next we show that $B$ is a minimal DFA.

(I) Every state $I\in Q_B$ is reachable from $\{m-1\}$. The proof is
similar to that of  Theorem~\ref{L_1^R L_2 lower bound}. We consider
the following four cases:

\begin{itemize}
\item[{\rm 1.}]$I=\emptyset$.
$\delta_A(\{ m-1 \}, b) = I=\emptyset .$
\item[{\rm 2.}]$I=f_B$. $\delta_A(\{ m-1 \}, a^{m-1}) =I=f_B.$
\item[{\rm 3.}]$|I|=1$.
Assume that $I=\{ i \}$, $1\leq i\leq m-1$. Note that $i\neq 0$
because all the final states in $A$ have been merged into $f_B$. In
this case, $\delta_A(\{ m-1 \}, a^{m-1-i}) =I.$
\item[{\rm 4.}]$2\leq |I|\leq m$.
Assume that $I=\{ i_1, i_2, \ldots ,i_k \}$, $1\leq i_1<i_2< \ldots
<i_k \leq m-1$, $2\leq k\leq m$. $\delta_A(\{ m-1 \}, w) = I$, where
$$w = ab(ac)^{i_2-i_1-1}ab(ac)^{i_{3}-i_{2}-1}\cdots ab(ac)^{i_k-i_{k-1}-1}a^{m-1-i_k}.$$
\end{itemize}

(II) Any two different states $I$ and $J$ in $Q_B$ are
distinguishable.

Since $f_B$ is the only final state in $Q_B$, it is inequivalent to
any other state. Thus, we consider the case when neither of $I$ and
$J$ is $f_B$.

Without loss of generality, we may assume that $|I|\geq |J|$. Let
$x\in I-J$. $x$ is always greater than $0$ because all the states
which include $0$ have been merged into $f_B$. Then a string
$d^{x-1}a$ can distinguish these two states because
\begin{eqnarray*}
\delta_B(I, d^{x-1}a)& = & f_B,\\
\delta_B(J, d^{x-1}a) & \neq & f_B.
\end{eqnarray*}
Since all the states in $B$ are reachable and pairwise
distinguishable, $B$ is a minimal DFA. Thus, any DFA accepting
$L(M))^R\Sigma^*$ needs at least $2^{m-1}+1$ states.
\end{proof}

After summarizing Theorem~\ref{L_1^R L_2 upper bound m>=2 n=1},
Theorem~\ref{L_1^R L_2 lower bound m>=4 n=1} and Lemma~\ref{L_1^R
L_2 lower bound m=2 or 3 n=1}, we obtain the state complexity of the
combined operation $L_1^RL_2$ for $m \ge 2$ and $n=1$.

\begin{theorem}
\label{L_1^R L_2 state complexity m>=2 n=1} For any integer $m \ge
2$, let $L_1$ be an $m$-state DFA language and $L_2$ be a $1$-state
DFA language. Then $2^{m-1}+1$ states are both sufficient and
necessary in the worst case for a DFA to accept $L_1^RL_2$.
\end{theorem}

\section{Star combined with catenation}\label{sec:star-cat}
In this section, we investigate the state complexity of $L(A)^*L(B)$ for two DFAs $A$ and $B$ of sizes $m,n \ge 1$, respectively.
We first notice that, when $n = 1$, the state complexity of $L(A)^*L(B)$ is 1 for any $m \ge 1$.
This is because $B$ is complete ($L(B)$ is either $\emptyset$ or $\Sigma^*$), and we have either $L(A)^*L(B) = \emptyset$ or $\Sigma^* \subseteq L(A)^*L(B) \subseteq \Sigma^*$.
Thus, $L(A)^*L(B)$ is always accepted by a 1 state DFA.
Next, we consider the case where $A$ has only one final state and it is also the initial state.
In such a case, $L(A)^*$ is also accepted by $A$, and hence the state complexity of $L(A)^*L(B)$ is equal to that of $L(A)L(B)$.
We will show that, for any $A$ of size $m \ge 1$ in this form and any $B$ of size $n \ge 2$, the state complexity of $L(A)L(B)$ (also $L(A)^*L(B)$) is $m(2^n-1) - 2^{n-1} + 1$ (Theorems~\ref{thm:star-cat-upper-special} and~\ref{thm:star-cat-lower-special}), which is lower than the state complexity of catenation in the general case.
Lastly, we consider the state complexity of $L(A)^*L(B)$ in the remaining case, that is when $A$ has at least a final state that is not the initial state and $n \ge 2$.
We will show that its upper bound (Theorem~\ref{thm:star-cat-upper}) coincides with its lower bound (Theorem~\ref{thm:star-cat-lower}), and the state complexity is $5 \cdot 2^{m+n-3} - 2^{m-1} - 2^n +1$.

Now, we consider the case where DFA $A$ has only one final state and it is also the initial state, and first obtain the following upper bound of the state complexity of $L(A)L(B)$ ($L(A)^*L(B)$), for any DFA $B$ of size $n \ge 2$.
\begin{theorem}\label{thm:star-cat-upper-special}
For integers $m \ge 1$ and $n \ge 2$, let $A$ and $B$ be two DFAs with $m$ and $n$ states, respectively, where $A$ has only one final state and it is also the initial state.
Then, there exists a DFA of at most $m(2^n-1) - 2^{n-1} + 1$ states that accepts $L(A)L(B)$, which is equal to $L(A)^*L(B)$.
\end{theorem}
\begin{proof}
Let $A = (Q_1, \Sigma, \delta_1, s_1, \{s_1\})$ and $B = (Q_2, \Sigma, \delta_2, s_2, F_2)$.
We construct a DFA $C = (Q, \Sigma, \delta, s, F)$ such that
\begin{eqnarray*}
    & & Q = Q_1 \times ( 2^{Q_2} - \{ \emptyset \}) - \{s_1\} \times ( 2^{Q_2 - \{s_2\}} - \{\emptyset\}), \\
    & & s = \langle s_1, \{ s_2 \} \rangle, \\
    & & F = \{ \langle q, T \rangle \in Q \mid T \cap F_2 \neq \emptyset \}, \\
    & & \delta(\langle q, T \rangle, a) = \langle q', T' \rangle,
            \mbox{ for $a \in \Sigma$, where $q' = \delta_1(q, a)$ and $T' = R \cup \{s_2\}$} \\
    & & \hspace{2cm}  \mbox{if $q' = s_1$, $T' = R$ otherwise, where $R = \delta_2(T,a)$.}
\end{eqnarray*}
Intuitively, $Q$ contains the pairs whose first component is a state of $Q_1$ and second component is a subset of $Q_2$.
Since $s_1$ is the final state of $A$, without reading any letter, we can enter the initial state of $B$.
Thus, states $\langle q, \emptyset \rangle$ such that $q \in Q_1$ can never be reached in $C$, because $B$ is complete.
Moreover, $Q$ does not contain those states whose first component is $s_1$ and second component does not contain $s_2$.

Clearly, $C$ has $m(2^n-1) - 2^{n-1} + 1$ states, and we can verify that $L(C) = L(A)L(B)$.
\end{proof}

Next, we show that this upper bound can be reached by some witness DFAs in the specific form.

\begin{figure}[ht]
  \begin{center}
  \includegraphics[scale=0.17]{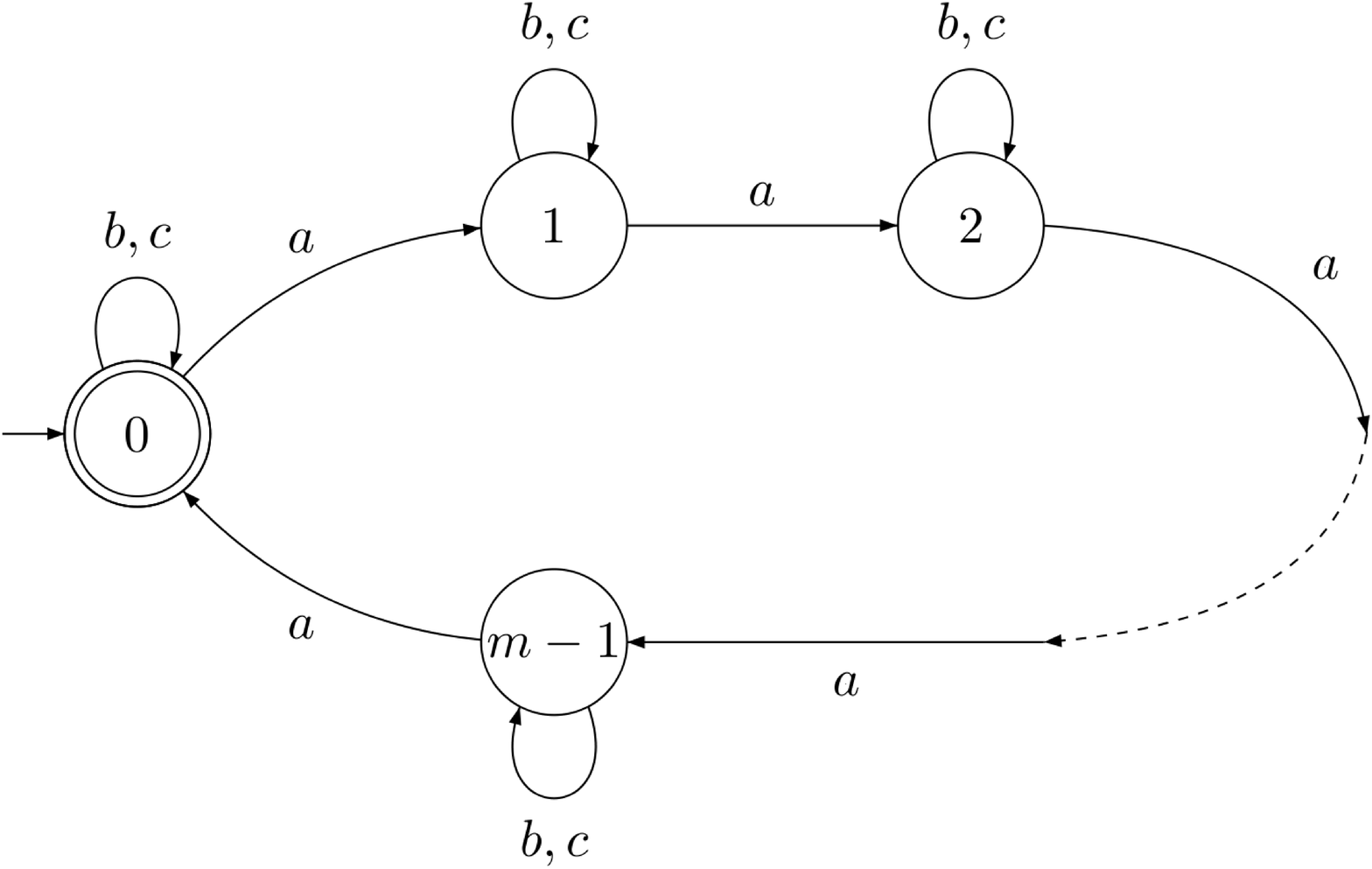}
  \end{center}
  \caption{Witness DFA $A$ for Theorem~\ref{thm:star-cat-lower-special} when $m \ge 2$}
\label{fig:DFAA-star-cat-special}
\end{figure}

\begin{figure}[ht]
  \begin{center}
  \includegraphics[scale=0.17]{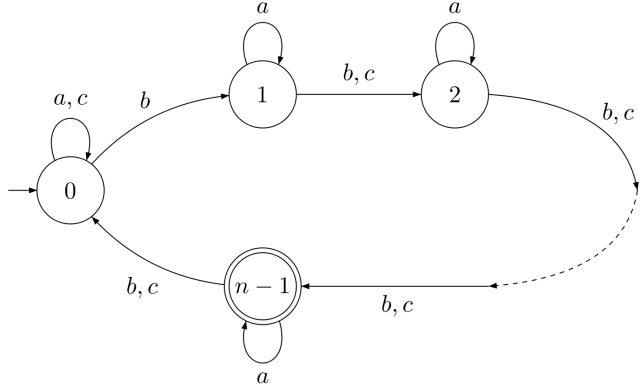}
  \end{center}
  \caption{Witness DFA $B$ for Theorem~\ref{thm:star-cat-lower-special} when $m \ge 2$}
\label{fig:DFAB-star-cat-special}
\end{figure}

\begin{theorem}\label{thm:star-cat-lower-special}
For any integers $m \ge 1$ and $n \ge 2$, there exist a DFA $A$ of $m$ states and a DFA $B$ of $n$ states, where $A$ has only one final state and it is also the initial state, such that any DFA accepting the language $L(A)L(B)$, which is equal to $L(A)^*L(B)$, needs at least $m(2^n-1) - 2^{n-1} + 1$ states.
\end{theorem}
\begin{proof}
When $m = 1$, the witness DFAs used in the proof of Theorem 1 in~\cite{YuZhSa94} can be used to show that the upper bound proposed in Theorem~\ref{thm:star-cat-upper-special} can be reached.

Next, we consider the case when $m \ge 2$.
We provide witness DFAs $A$ and $B$, depicted in Figures~\ref{fig:DFAA-star-cat-special} and~\ref{fig:DFAB-star-cat-special}, respectively, over the three letter alphabet $\Sigma = \{a,b,c\}$.

$A$ is defined as $A = (Q_1, \Sigma, \delta_1, 0 , \{0\})$ where $Q_1 = \{0,1,\ldots,m-1\}$, and the transitions are given as
    \begin{itemize}
    \item $\delta_1 (i, a) = i+1 \mbox{ mod } m$, for $i \in Q_1$,
    \item $\delta_1 (i, x) = i$, for $i \in Q_1$, where $x \in \{b,c\}$.
    \end{itemize}

$B$ is defined as $B = (Q_2, \Sigma, \delta_2, 0, \{n-1\})$ where $Q_2 = \{0,1,\ldots,n-1\}$, where the transitions are given as
    \begin{itemize}
    \item $\delta_2 (i, a) = i$, for $i \in Q_2$,
    \item $\delta_2 (i, b) = i+1 \mbox{ mod } n$, for $i \in Q_2$,
    \item $\delta_2 (0, c) = 0$, $\delta_2 (i, c) = i+1 \mbox{ mod } n$, for $i \in \{1, \ldots, n-1\}$.
    \end{itemize}

    Following the construction described in the proof of Theorem~\ref{thm:star-cat-upper-special}, we construct a DFA $C = (Q, \Sigma, \delta, s, F)$ that accepts $L(A)L(B)$ (also $L(A)^*L(B)$).
    To prove that $C$ is minimal, we show that (I) all the states in $Q$ are reachable from $s$, and (II) any two different states in $Q$ are not equivalent.

    For (I), we show that all the state in $Q$ are reachable by induction on the size of $T$.

    The basis clearly holds, since, for any $i \in Q_1$, state $\langle i, \{0\} \rangle$ is reachable from $\langle 0, \{0\} \rangle$ by reading string $a^{i}$, and state $\langle i, \{j\} \rangle$ can be reached from state $\langle i, \{0\} \rangle$ on string $b^{j}$, for any $i \in \{1, \ldots, m-1\}$ and $j \in Q_2$.

    In the induction steps, we assume that all the states $\langle q, T \rangle$ such that $|T| < k$ are reachable.
    Then, we consider the states $\langle q, T \rangle$ where $|T| = k$.
    Let $T = \{j_1, j_2, \ldots, j_k\}$ such that $0 \le j_1 < j_2 < \ldots < j_k \le n-1$.
    We consider the following three cases:
    \begin{enumerate}
    \item $j_1 = 0$ and $j_2 = 1$.
    For any state $i \in Q_1$, state $\langle i , T \rangle \in Q$ can be reached as
    \[
    \langle i, \{0, 1, j_3, \ldots, j_k\} \rangle = \delta(\langle 0, \{0, j_3 - 1, \ldots, j_k - 1\}\rangle, ba^{i}),
    \]
    where $\{0, j_3 - 1, \ldots, j_k - 1\}$ is of size $k-1$.

    \item $j_1 = 0$ and $j_2 > 1$.
    For any state $i \in Q_1$, state $\langle i, \{0, j_2, \ldots, j_k\} \rangle$ can be reached from state $\langle i, \{0, 1, j_3 - j_2 + 1, \ldots, j_k - j_2 + 1\} \rangle$ by reading string $c^{j_2 - 1}$.

    \item $j_1 > 0$.
    In such a case, the first component of state $\langle q, T \rangle$ cannot be $0$.
    Thus, for any state $i \in \{1, \ldots, m-1\}$, state $\langle i, \{j_1, j_2, \ldots, j_k\} \rangle$ can be reached from state $\langle i, \{0, j_2 - j_1, \ldots, j_k - j_1\} \rangle$ by reading string $b^{j_1}$.
    \end{enumerate}

    Next, we show that any two distinct states $\langle q, T \rangle$ and $\langle q', T' \rangle$ in $Q$ are not equivalent.
    We consider the following two cases:
    \begin{enumerate}
    \item $q \neq q'$.
    Without loss of generality, we assume $ q \neq 0$.
    Then, string $w = c^{n-1}a^{m-q}b^{n}$ can distinguish the two states, since $\delta(\langle q, T \rangle, w) \in F$ and $\delta(\langle q', T' \rangle, w) \not\in F$.

    \item $q = q'$ and $T \neq T'$.
    Without loss of generality, we assume that $|T| \ge |T'|$.
    Then, there exists a state $j \in T - T'$.
    It is clear that, when $q \neq 0$, string $b^{n-1-j}$ can distinguish the two states, and when $q = 0$, string $c^{n-1-j}$ can distinguish the two states since $j$ cannot be $0$.
    \end{enumerate}

    Due to (I) and (II), DFA $C$ needs at least $m(2^n-1) - 2^{n-1} + 1$ states and is minimal.
\end{proof}

In the rest of this section, we focus on the case where DFA $A$ contains at least one final state that is not the initial state.
Thus, this DFA is of size at least 2.
We first obtain the following upper bound for the state complexity.

\begin{theorem}\label{thm:star-cat-upper}
Let $A = (Q_1, \Sigma, \delta_1, s_1, F_1)$ be a DFA such that $|Q_1| = m > 1$ and $|F_1 - \{s_1\}| = k_1 \ge 1$, and $B = (Q_2, \Sigma, \delta_2, s_2, F_2)$ be a DFA such that $|Q_2| = n > 1$.
Then, there exists a DFA of at most $(\dfrac{3}{4}2^m - 1)(2^n-1) - (2^{m-1}-2^{m-k_1-1})(2^{n-1}-1)$ states that accepts $L(A)^* L(B)$.
\end{theorem}
\begin{proof}
We denote $F_1 - \{s_1\}$ by $F_0$. Then, $|F_0| = k_1 \ge 1$.

We construct a DFA $C = \{Q, \Sigma, \delta, s, F\}$ for the language $L_1^* L_2$, where $L_1$ and $L_2$ are the languages accepted by DFAs $A$ and $B$, respectively.

Let $Q = \{ \langle p, t \rangle \mid p \in P \mbox{ and } t \in T\} - \{\langle p', t' \rangle \mid p' \in P' \mbox{ and } t' \in T'\}$, where
\begin{eqnarray*}
    P & = & \{ R \mid R \subseteq (Q_1 - F_0) \mbox{ and } R \neq \emptyset\} \cup \{ R \mid R \subseteq Q_1, s_1 \in R, \mbox{ and } R \cap F_0 \neq \emptyset \},\\
    T & = &  2^{Q_2} - \{\emptyset\}, \\
    P' & = & \{ R \mid R \subseteq Q_1, s_1 \in R, \mbox{ and } R \cap F_0 \neq \emptyset \}, \\
    T' & = & 2^{Q_2 - \{s_2\}}- \{\emptyset\}.
\end{eqnarray*}

The initial state $s$ is $s = \langle \{s_1\}, \{s_2\}\rangle$.

The set of final states is defined to be $F = \{ \langle p, t\rangle \in Q \mid t \cap F_2 \neq \emptyset\}$.

The transition relation $\delta$ is defined as follows:
\[
\delta (\langle p,t \rangle, a) = \left\{
\begin{array}{l l}
  \langle p', t'\rangle & \quad \text{if $p' \cap F_1 = \emptyset$,}\\
  \langle p', t' \cup \{s_2\} \rangle & \quad \text{otherwise,}\\
\end{array} \right.
\]
where, $a \in \Sigma$, $p' = \delta_1(p, a)$, and $t' = \delta_2(t, a)$.

    Intuitively, $C$ is equivalent to the NFA $C'$ obtained by first constructing an NFA $A'$ that accepts $L_1^*$, then catenating this new NFA with DFA $B$ by $\lambda$-transitions.
    Note that, in the construction of $A'$, we need to add a new initial and final state $s_1'$.
    However, this new state does not appear in the first component of any of the states in $Q$.
    The reason is as follows.
    First, note that this new state does not have any incoming transitions.
    Thus, from the initial state $s_1'$ of $A'$, after reading a nonempty word, we will never return to this state.
    As a result, states $\langle p, t \rangle$ such that $p \subseteq Q_1 \cup \{s_1'\}$, $ s_1' \in p$, and $t \in 2^{Q_2}$ is never reached in DFA $C$ except for the state $\langle\{s_1'\}, \{s_2\}\rangle$.
    Then, we note that, in the construction of $A'$, states $s_1'$ and $s_1$ should reach the same state on any letter in $\Sigma$.
    Thus, we can say that states $\langle \{s_1'\}, \{s_2\}\rangle$ and $\langle \{s_1\}, \{s_2\} \rangle$ are equivalent, because either of them is final if $s_2 \not\in F_2$, and they are both final states otherwise.
    Hence, we merge this two states and let $\langle \{s_1\}, \{s_2\} \rangle$ be the initial state of $C$.

    Also, we notice that states $\langle p, \emptyset \rangle$ such that $p \in P $ can never be reached in $C$, because $B$ is complete.

    Moreover, $C$ does not contain those states whose first component contains a final state of $A$ and whose second component does not contain the initial state of $B$.

    Therefore, we can verify that DFA $C$ indeed accepts $L_1^* L_2$, and it is clear that the size of $Q$ is
 \[
    (\dfrac{3}{4}2^m - 1)(2^n-1) - (2^{m-1}-2^{m-k_1-1})(2^{n-1}-1).
 \]
\end{proof}

Then, we show that this upper bound is reachable by some witness DFAs.

\begin{figure}[ht]
  \begin{center}
  \includegraphics[scale=0.17]{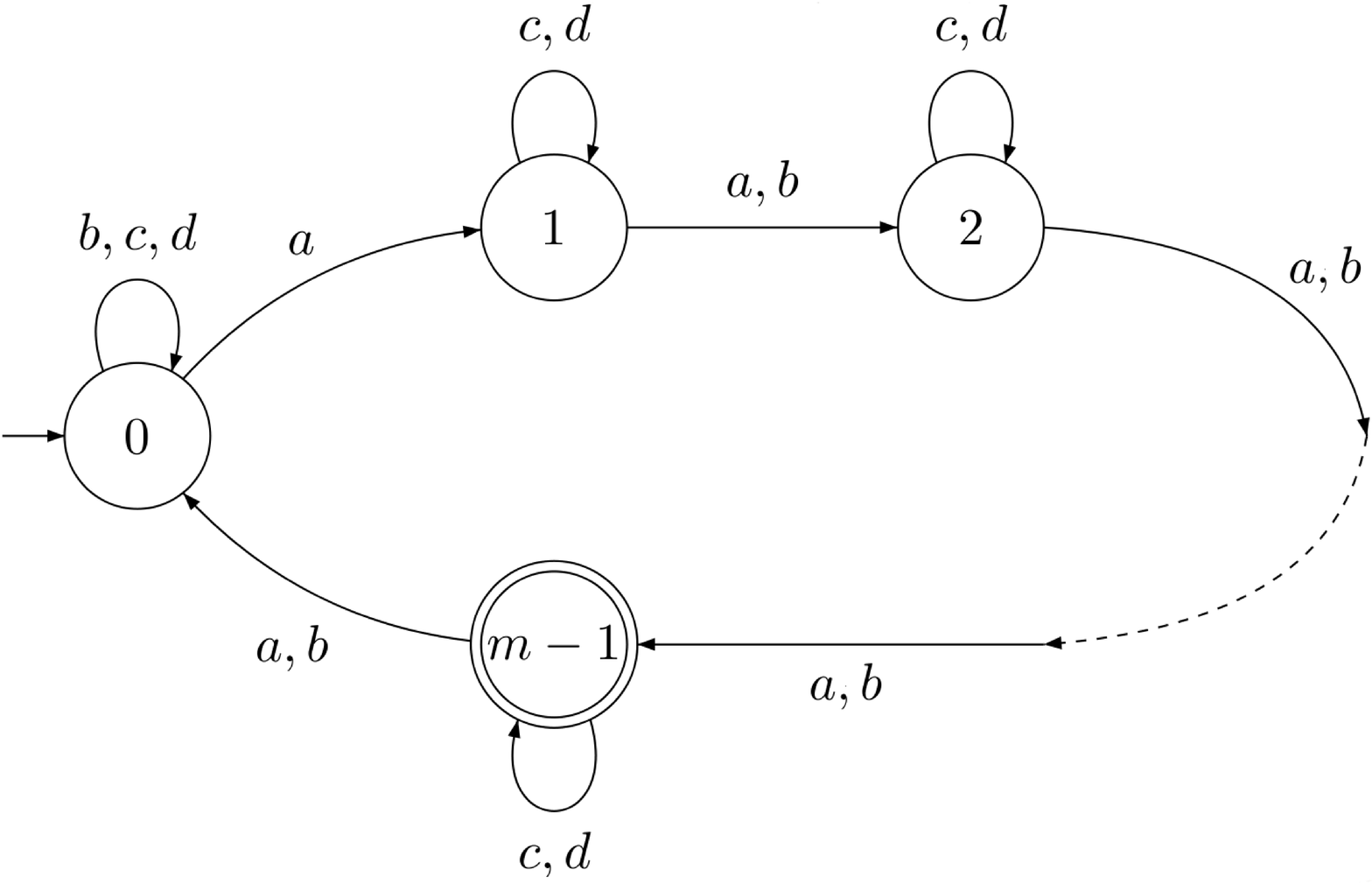}
  \end{center}
  \caption{Witness DFA $A$ for Theorem~\ref{thm:star-cat-lower}}
\label{fig:DFAA-star-cat}
\end{figure}

\begin{figure}[ht]
  \begin{center}
  \includegraphics[scale=0.17]{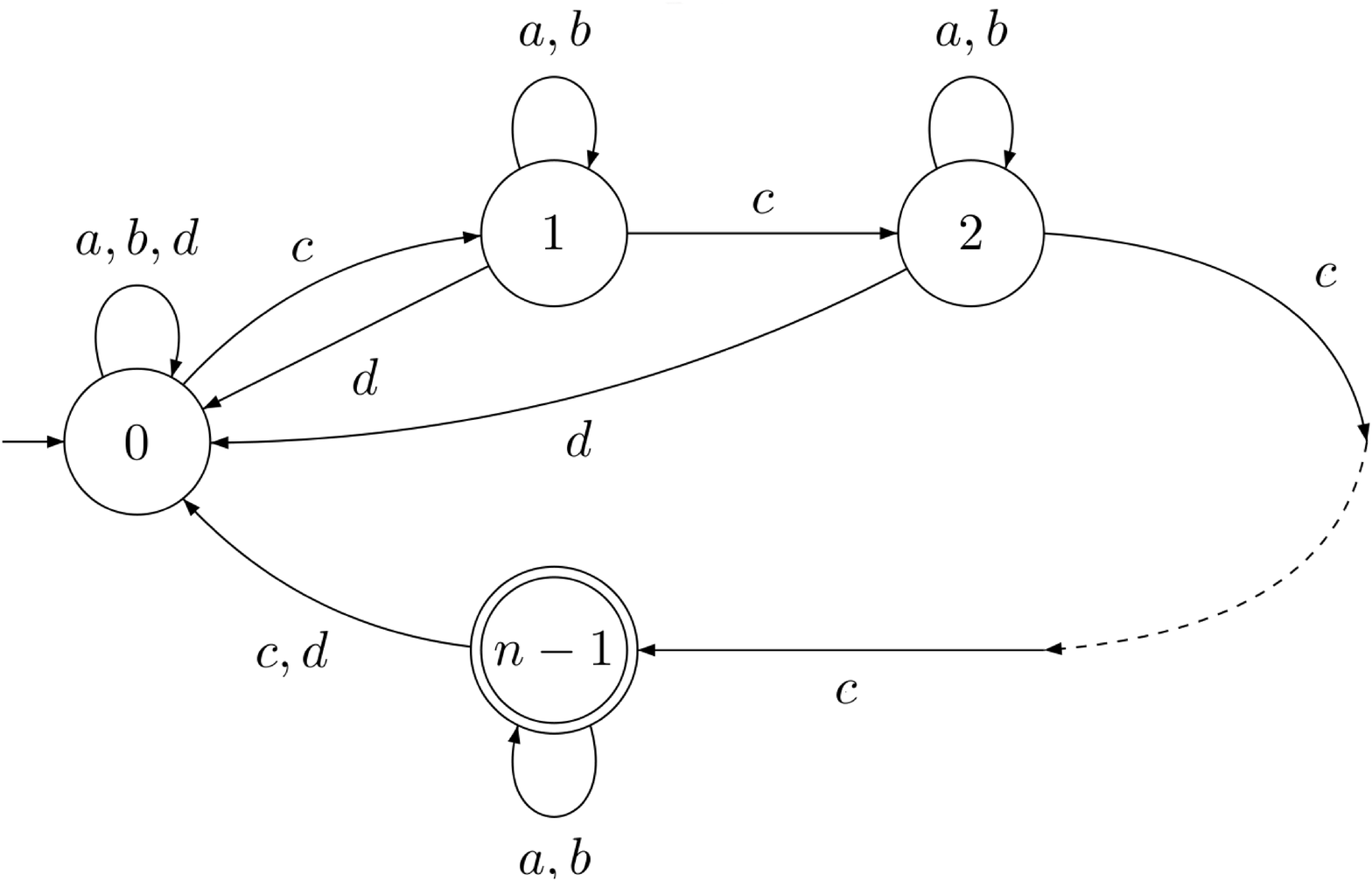}
  \end{center}
  \caption{Witness DFA $B$ for Theorem~\ref{thm:star-cat-lower}}
\label{fig:DFAB-star-cat}
\end{figure}

\begin{theorem}\label{thm:star-cat-lower}
For any integers $m,n \ge 2$, there exist a DFA $A$ of $m$ states and a DFA $B$ of $n$ states such that any DFA accepting $L(A)^* L(B)$ needs at least $5 \cdot 2^{m+n-3} - 2^{m-1} - 2^n +1$ states.
\end{theorem}
\begin{proof}
We define the following two automata over a four letter alphabet $\Sigma = \{a,b,c,d\}$.

Let $A = (Q_1, \Sigma, \delta_1, 0, \{m-1\})$, shown in Figure~\ref{fig:DFAA-star-cat}, where $Q_1 = \{0,1,\ldots,m-1\}$, and the transitions are defined as
    \begin{itemize}
    \item $\delta_1 (i, a) = i+1 \mbox{ mod } m$, for $i \in Q_1$,
    \item $\delta_1 (0, b) = 0$, $\delta_1 (i, b) = i+1 \mbox{ mod } m$, for $i \in \{1, \ldots, m-1\}$,
    \item $\delta_1 (i, x) = i$, for $i \in Q_1$, $x \in \{c,d\}$.
    \end{itemize}

Let $B = (Q_2, \Sigma, \delta_2, 0, \{n-1\})$, shown in Figure~\ref{fig:DFAB-star-cat}, where $Q_2 = \{0,1,\ldots,n-1\}$, and the transitions are defined as
    \begin{itemize}
    \item $\delta_2 (i, x) = i$, for $i \in Q_2$, $x \in \{a,b\}$,
    \item $\delta_2 (i, c) = i+1 \mbox{ mod } n$, for $i \in Q_2$,
    \item $\delta_2 (i, d) = 0$, for $i \in Q_2$.
    \end{itemize}

    Let $C = \{Q, \Sigma, \delta, \langle \{0\}, \{0\} \rangle, F\}$ be the DFA accepting the language $L(A)^*L(B)$ which is constructed from $A$ and $B$ exactly as described in the proof of Theorem~\ref{thm:star-cat-upper}.

    Now, we prove that the size of $Q$ is minimal by showing that (I) any state in $Q$ can be reached from the initial state, and (II) no two different states in $Q$ are equivalent.

    We first prove (I) by induction on the size of the second component $t$ of the states in $Q$.

    {\bf Basis:} for any $i \in Q_2$, state $\langle \{0\}, \{i\} \rangle$ can be reached from the initial state $\langle \{0\}, \{0\}\rangle$ on string $c^i$.
    Then, by the proof of Theorem 5 in~\cite{YuZhSa94}, it is clear that state $\langle p, \{i\}\rangle$ of $Q$, where $p \in P$ and $i \in Q_2$, is reachable from state $\langle \{0\}, \{i\}\rangle$ on strings over letters $a$ and $b$.

    {\bf Induction step:} assume that all the states $\langle p, t\rangle$ in $Q$ such that $p \in P$ and $|t| < k$ are reachable.
    Then, we consider the states $\langle p, t\rangle$ in $Q$ where $p \in P$ and $|t| = k$.
    Let $t = \{j_1, j_2, \ldots, j_k\}$ such that $0 \le j_1 < j_2 < \ldots < j_k \le n-1$.

    Note that states such that $p = \{0\}$ and $j_1 = 0$ are reachable as follows:
\[
    \langle \{0\}, \{0, j_2, \ldots, j_k\}\rangle = \delta( \langle \{0\}, \{0, j_3 - j_2, \ldots, j_k - j_2 \}\rangle, c^{j_2}a^{m-1}b).
\]
Then, states such that $p = \{0\}$ and $j_1 > 0$ can be reached as follows:
\[
    \langle \{0\}, \{j_1, j_2, \ldots, j_k\}\rangle = \delta(\langle \{0\}, \{0, j_2 - j_1, \ldots, j_k - j_1\}\rangle, c^{j_1}).
\]

Once again, by using the proof of Theorem 5 in~\cite{YuZhSa94}, states $\langle p, t\rangle$ in $Q$, where $p \in P$ and $|t| = k$, can be reached from the state $\langle \{0\}, t \rangle$ on strings over letters $a$ and $b$.

Next, we show that any two states in $Q$ are not equivalent.
Let $\langle p, t\rangle$ and $\langle p', t' \rangle$ be two different states in $Q$.
We consider the following two cases:
\begin{enumerate}
\item $p \neq p'$.
Without loss of generality, we assume $|p| \ge |p'|$.
Then, there exists a state $ i \in p - p'$.
It is clear that string $a^{m-1-i}dc^n$ is accepted by $C$ starting from state $\langle p, t\rangle$, but it is not accepted starting from state $\langle p', t' \rangle$.

\item $p = p'$ and $t \neq t'$.
We may assume that $|t| \ge |t'|$ and let $j \in t-t'$.
Then, state $\langle p, t \rangle$ reaches a final state on string $c^{n-1-j}$, but state $\langle p', t' \rangle$ does not on the same string.
Note that, when $m-1 \in p$, we can say that $j \neq 0$.
\end{enumerate}

Due to (I) and (II), DFA $C$ has at least $5 \cdot 2^{m+n-3} - 2^{m-1} - 2^n +1$ reachable states, and any two of them are not equivalent.
\end{proof}

\section{Conclusion}\label{sec:conclusion}
In this paper, we have studied the state complexities of two combined operations: reversal combined with catenation and star combined with catenation.
We showed that, due to the structural properties of DFAs obtained from reversal and star, the state complexities of these two combined operations are not equal but close to the mathematical compositions of the state complexities of their individual participating operations.

\end{document}